\documentclass[reqno]{amsart}

\usepackage{amssymb}

\usepackage[dvips,colorlinks=true,linkcolor=blue]{hyperref}

\newcounter{numcount}
\newcommand{\labelnummer}{\mbox{\normalfont (\roman{numcount})}}%

\makeatletter

   {\let\curlabelspeicher\@currentlabel%
     \begin{list}{\labelnummer}%
       {\usecounter{numcount}\leftmargin0pt%
         \topsep0.5ex\partopsep2ex\parsep0pt\itemsep0ex\@plus1\p@%
         \labelwidth2.5em\itemindent3.5em\labelsep1em%
       }%
     \let\saveitem\item%
     \def\item{\saveitem%
  \def\@currentlabel{\curlabelspeicher\hskip0.5pt\labelnummer}}%
     \let\savelabel\label%
     \def\label##1{\savelabel{##1}%
       \@bsphack%
         \ifmmode\else%
           \protected@write\@auxout{}%
           {\string\newlabel{##1item}{{\labelnummer}{\thepage}}}%
         \fi%
       \@esphack%
     }%
   }{\end{list}}%

\numberwithin{equation}{section}

\newcommand{\pro}{\mathbb{P}}

\newcommand{\vers}{\operatornamewithlimits{\to}}
\newcommand{\car}{\mathbf{1}}

\theoremstyle{plain}
\newtheorem{Th}{Theorem}[section]
\newtheorem{Le}{Lemma}[section]

\theoremstyle{definition}
\newtheorem{Rem}{Remark}[section]

\DeclareMathOperator{\supp}{supp}
\DeclareMathOperator{\tr}{tr}

\DeclareMathOperator{\dist}{dist}

\newcommand\R{\mathbb R}

\newcommand\Z{\mathbb Z}

\renewcommand\P{\mathbb P}
\newcommand\E{\mathbb E}

\newcommand{\T}{\mathbb{T}}

\newcommand\e{\mathrm{e}}

\newcommand\eps{\varepsilon}

\newcommand\beq{\begin{equation}}
\newcommand\eeq{\end{equation}}

\begin{document}

\title[Enhanced Wegner and Minami estimates and applications]
{Enhanced Wegner and Minami estimates and eigenvalue statistics of random Anderson models at spectral edges}

\author{Fran\c cois Germinet} \address{Universit{\'e} de Cergy-Pontoise,
  CNRS UMR 8088, IUF, D{\'e}partement de math{\'e}matiques, F-95000
  Cergy-Pontoise, France}
\email{\href{mailto:francois.germinet@u-cergy.fr}{francois.germinet@u-cergy.fr}}

\author{Fr{\'e}d{\'e}ric Klopp} \address{Institut de Mathmatiques de Jussieu,
  Universit{\'e} Pierre et Marie Curie, Case 186, 4 place Jussieu F-75252
  Paris cedex 05, FRANCE}
\email{\href{mailto:klopp@math.jussieu.fr}{klopp@math.jussieu.fr}}


\thanks{The authors are supported by the grant ANR-08-BLAN-0261-01.
  The authors would also like to thank the Centre Interfacultaire
  Bernoulli (EPFL, Lausanne) for its hospitality.}


\begin{abstract}
  We consider the discrete Anderson model and prove enhanced Wegner
  and Minami estimates where the interval length is replaced by the
  IDS computed on the interval. We use these estimates to improve on
  the description of finite volume eigenvalues and eigenfunctions
  obtained in \cite{Ge-Kl:10}. As a consequence of the improved
  description of eigenvalues and eigenfunctions, we revisit a number
  of results on the spectral statistics in the localized regime
  obtained in~\cite{Ge-Kl:10,Kl:10a} and extend their domain of
  validity, namely :
  \begin{itemize}
  \item the local spectral statistics for the unfolded eigenvalues;
  \item the local asymptotic ergodicity of the unfolded eigenvalues;
  \end{itemize}
  In dimension $1$, for the standard Anderson model, the improvement
  enables us to obtain the local spectral statistics at band edge,
  that is in the Lifshitz tail regime. In higher dimensions, this
  works for modified Anderson models.
\end{abstract}

\setcounter{section}{0}
\maketitle

\section{Introduction}
Anderson models are known to exhibit a region of localized states,
either at the edges of the spectrum, or in a given range of energies
if the disorder is large enough. Within this region of localization it
is natural to study the two basic components of the spectral theory in
this case: the eigenfunctions (how localized they are, etc), and the
eigenvalues (their multiplicity, their statistics, etc). The
localization properties of eigenfunctions in the localized phase are
by now quite well understood
(e.g. \cite{MR2509110,MR1859605,MR2002h:82051,MR2203782}). The precise
description of these localization properties plays an important role
in the understanding of many physical phenomena (e.g. dynamical
localization, constancy of the Hall conductance in quantum Hall
systems, Mott formula). Much less works have been devoted to the
understanding of the eigenvalues statistics of the Anderson model;
while for random matrices many more results are available
(e.g. \cite{PhysRevLett.75.69, MR1793333,MR2124079}). Poisson
statistics of Anderson eigenvalues
have been studied in \cite{MR84e:34081,MR97d:82046,MR2663411}.\\
An important ingredient that enters the proof of localization is a so
called Wegner estimate that controls the probability of finding an
eigenvalue of the finite volume operator in a small interval of
energy.  As for Poisson statistics, the analysis of
\cite{MR97d:82046,MR2663411} relies on localization properties, the
Wegner estimate and on a so called Minami estimate that enables one to
control the probability of the occurrence of two eigenvalues of the
finite volume operator in a small interval of energy. It is worth
mentioning that simplicity of the spectrum is a direct consequence of
localization properties combined with a Minami estimate \cite{KM06}.\\
Recently, in \cite{Ge-Kl:10}, the present authors introduced a refined
way of describing eigenvalues and centers of localization in finite
volumes, through a reduction procedure that enables one to approximate
eigenvalues at finite volume by independent and identically
distributed random variables. With this reduction in hand, they could
in particular extend known results about Poisson statistics and obtain
the first asymptotic result for the eigenlevel spacings
distribution. \cite{Kl:10a} used this reduction to study the local
ergodicity of the unfolded eigenvalues.\\
In the present article, we introduce enhanced Wegner and Minami
estimates that we only only within the region of localization. The
main novelty is that they take into account the weight that the
integrated density of states (IDS) gives to intervals to estimate the
probabilities of the occurrence of a single or of multiple eigenvalues
in a small energy interval. These estimates enable us to revisit the reduction procedure mentioned above and get better controls. We thus remove some limitations of \cite{Ge-Kl:10} and cover situations where the IDS
gets too small for the analysis of \cite{Ge-Kl:10} to be valid: this
happens when the IDS is exponentially small in an inverse power of the
length of the interval. As an application, our results enable us to
prove Poisson statistics for the unfolded eigenvalues in dimension 1
at the band edges, that is where a Lifshitz tail regime
occurs. To our best knowledge, this is the first such result.\\
As another application, we provide improved large deviation estimates
for the number of finite volume eigenvalues contained in suitably
scaled intervals, as well as a central limit theorem for this
quantity.
\section{Main results}
\label{sec:main-results}
We consider the discrete Anderson Hamiltonian
\begin{equation}
  \label{AndH} 
  H_{\omega}: = H_0 + V_{\omega} ,
\end{equation}
acting on $\ell^2(\mathbb{Z}^d)$, where
\begin{itemize}
\item $H_0$ is a convolution matrix with exponentially decaying
  off-diagonal coefficients i.e. exponential off-diagonal decay that
  is $H_0=((h_{k-k'}))_{k,k'\in \Z^d}$ such that,
  \begin{itemize}
  \item $h_{-k}=\overline{h_k}$ for $k\in\Z^d$ and for some $k\not=0$,
    $h_k\not=0$.
  \item there exists $c>0$ such that, for $k\in\Z^d$,
    \begin{equation}
      \vert h_k\vert\leq \frac1ce^{-c\vert k\vert}.
    \end{equation}
  \end{itemize}
  Define
  \begin{equation}
    \label{eq:19}
    h(\theta)=\sum_{k\in\Z^d}h_ke^{ik\theta}\text{ where
    }\theta=(\theta_1,\dots,\theta_d)\in\R^d.
  \end{equation}
\item $V_{\omega}$ is an Anderson potential:
  \begin{equation}
    \label{AndV} V_{\omega} (x):= \sum_{j \in \Z^d} \omega_j
    \Pi_j .  
  \end{equation}
\end{itemize}
Where $ \Pi_j$ is the projection onto site $j$, and $\omega=\{
\omega_j \}_{j\in \Z^d}$ is a family of independent identically
distributed random variables whose common probability
distribution $\mu$ is non-degenerate and has a bounded density $g$.\\
We denote by
\begin{itemize}
\item $\Sigma\subset\R$ the almost sure spectrum of $H_\omega$ (see
  e.g.~\cite{MR2509110,MR94h:47068}); it is known that
  $\Sigma=h(\R^d)+$supp$\,g$;
\item $\Sigma_{SDL}\subset\Sigma\subset\R$ the set of energies where
  strong dynamical localization holds; we refer to
  Theorem~\ref{thmFVE} for a precise description of $\Sigma_{SDL}$; it
  is known that such a region of energies exists at least near the
  edges of the spectrum $\Sigma$ (see e.g.~\cite{MR2509110,
    MR1859605,MR2002h:82051,MR2078370,MR2203782}).
\end{itemize}
Recall (see~\cite{MR2509110}) the integrated density of states (IDS)
may be defined as
\begin{equation}
  \label{defIDS}
  N(E) =\E \tr (\Pi_0\car_{]-\infty,E]}(H_{\omega}) \Pi_0)
  =\E\langle\delta_0,\car_{]-\infty,E]}(H_{\omega})\delta_0\rangle.
\end{equation}
In particular, if $I$ is an interval, we define $N(I)$ as
\begin{equation}
  \label{NI}
  N(I):=\E\tr (\Pi_0\car_{I}(H_{\omega})
  \Pi_0)=\E\langle\delta_0,\car_{I}(H_{\omega})\delta_0\rangle.
\end{equation}
For $L>1$, consider $\Lambda=[-L,L]^d\cap\Z^d$, a cube on the lattice
and let $H_\omega(\Lambda)$ be the random Hamiltonian $H_\omega$
restricted to $\Lambda$ with periodic boundary conditions.\\
Our analysis allows for the class of (continuous) random Schr{\"o}dinger
operators for which a Minami estimate has been derived, namely
Anderson type operators satisfying to a covering condition and single
site probabilities having a bounded density \cite{MR2663411,CGK11}.
\vskip.2cm
\noindent Notations: by $a\lesssim b$ we mean there exists a constant
$c\in]0,\infty[$ so that $a\le cb$. By $a\asymp b$ we mean there
exists a constant $c\in]1,\infty[$ such that $c^{-1} b \le a \le cb$.
\subsection{Improved versions of the Wegner and Minami estimates}
\label{sec:impr-vers-wegn}
We show Wegner and Minami estimates where the upper bounds keep trace
of the integrated density of states. In particular, they enable to
take advantage of the smallness of the integrated density when this
happens.\\
Let us first recall the usual Wegner and Minami estimates that are
known to hold for $H_\omega$ (see
e.g.~\cite{MR2509110,MR2360226,MR2290333,MR2505733} and
references therein):\\
\begin{description}
\item[(W)] $\displaystyle \E\left[\text{tr}(\car_J(H_\omega(\Lambda)))
  \right]\leq C |J|\,|\Lambda|$;\\
\item[(M)] $\displaystyle\E\left[\text{tr}(\car_J(H_\omega(\Lambda)))
    \cdot[\text{tr}(\car_J(H_\omega(\Lambda)))-1]\right]\leq C
  (|J|\,|\Lambda|)^2$.
\end{description}
\vskip.1cm\noindent Our main result is
\begin{Th}
  \label{thmWM}
  Fix $\xi\in(0,1)$. There exists constants $c,C\in(0,+\infty)$ such
  that for $L>1$ the following holds.
  \begin{enumerate}
  \item Let $I\subset \Sigma_{SDL}$ be a compact interval. Then
    \begin{equation}
      \label{Wcontrol}
      \left| \E \tr  \car_{I}(H_\omega(\Lambda))  - N(I)|\Lambda|
      \right| \le C \exp(-c L^\xi). 
    \end{equation}
    As a consequence, if $|N(I)|\ge C\exp(-c L^\xi)$ we get the Wegner
    estimate:
    \begin{equation}\label{W}
      \E( \tr \car_{I}(H_\omega(\Lambda))) \le 2 N(I)|\Lambda|.
    \end{equation}
  \item (High order Minami) Given $n\ge 2$ and $I_1\subset \cdots
    \subset I_n\subset \Sigma_{SDL}$ intervals so that $|N(I_n)|\ge
    C\exp(- c L^\xi)$,
    \begin{equation}
      \label{HOM}
      \E\left(\prod_{k=1}^n\left(\tr\car_{I_k}(H_\omega(\Lambda))-k+1
        \right)\right)\leq2\left(\prod_{k=1}^{n-1}\|\rho\|_\infty
        |I_k||\Lambda|\right) N(I_n)|\Lambda|.
    \end{equation}
    In particular, for $n=2$, if $|N(I)|\ge C\exp(-c L^\xi)$, we get
    the Minami estimate:
    \begin{equation}
      \label{M}
      \E\left[\tr\car_{I}(H_\omega(\Lambda)) (\tr
        \car_{I}(H_\omega(\Lambda)) -1 )\right]\leq2N(I)|I||\Lambda|^2.
    \end{equation}
  \end{enumerate}
\end{Th}

\begin{Rem}
  \label{rem:1}
 (i) The constant $2$ in~\eqref{W},~\eqref{HOM} and
\eqref{M} can be replaced by any constant larger than $1$, provided
$|\Lambda|$ is large enough.\\
(ii) The improved Minami estimates can also be proved for the continuous
  Anderson models considered in~\cite{MR2663411,CGK11} near the bottom of
  the spectrum.
\end{Rem}
\subsection{Local spectral statistics}
\label{sec:spetfral-statistics}
We shall combine Theorem~\ref{thmWM} with the description of the
eigenvalues of $H_{\omega}(\Lambda)$ obtained in the
paper~\cite{Ge-Kl:10} to obtain new results for the spectral
statistics of the Anderson model. In particular, the improved Minami
estimate enables us to remove the restriction on the smallness of the
density of states that was imposed in~\cite{Ge-Kl:10,Kl:10a} to obtain
results locally in energy.
\par Consider the eigenvalues of $H_\omega(\Lambda)$ ordered
increasingly and repeated according to multiplicity and denote them by
\begin{equation*}
  E_1(\omega,\Lambda)\leq E_2(\omega,\Lambda)\leq \cdots\leq
  E_{|\Lambda|}(\omega,\Lambda).  
\end{equation*}
Following~\cite{MR2352280,Mi:11}, define the {\it unfolded
  eigenvalues} as
\begin{equation*}
  0\leq N(E_1(\omega,\Lambda))\leq N(E_2(\omega,\Lambda))\leq
  \cdots\leq N(E_{|\Lambda|}(\omega,\Lambda))\leq 1.
\end{equation*}
Let $E_0$ be an energy in $\Sigma_{SDL}$.  The {\it unfolded local level
  statistics} near $E_0$ is the point process defined by
\begin{equation}
  \label{eq:13}
  \Xi(\xi; E_0,\omega,\Lambda) = 
  \sum_{j\geq1} \delta_{\xi_j(E_0,\omega,\Lambda)}(\xi),
\end{equation}
where
\begin{equation}
  \label{eq:11}
  \xi_j(E_0,\omega,\Lambda)=|\Lambda|(N(E_j(\omega,\Lambda))-N(E_0)).
\end{equation}
The unfolded local level statistics are described by
\begin{Th}
  \label{thr:3}
  Pick $E_0$ be an energy in $\Sigma_{SDL}$ such that the integrated
  density of states satisfies, for some $\rho\in(0,1/d)$,
  $\exists\,a_0>0$ s.t. $\forall a\in(-a_0,a_0)\cap(\Sigma-E_0)$,
  \begin{equation}
    \label{eq:60}
    |N(E_0+a)-N(E_0)|\geq e^{-|a|^{-\rho}}.
  \end{equation}
  When $|\Lambda|\to+\infty$, the point process
  $\Xi(E_0,\omega,\Lambda)$ converges weakly to
  \begin{itemize}
  \item a Poisson point process on the real line with intensity $1$ if
    $E_0\in\overset{\circ}{\Sigma}$, the interior of $\Sigma$.
  \item a Poisson point process on the half line with intensity $1$ if
    $E_0\in\partial\Sigma$, the half-line being $\R^+$ (resp. $\R^-$)
    if $(E_0-\varepsilon,E_0)\cap\Sigma=\emptyset$
    (resp. $(E_0,E_0+\varepsilon)\cap\Sigma=\emptyset$) for some
    $\varepsilon>0$.
  \end{itemize}
\end{Th}
\noindent The main improvement over~\cite[Theorem 1.2]{Ge-Kl:10} is
that the decay in assumption~\eqref{eq:60} can be taken exponential
(compare with~\cite[(1.12)]{Ge-Kl:10}); it does not depend anymore on
the Minami estimate.\\
In~\cite{Ge-Kl:10}, we state and prove stronger uniform results for
the convergence to Poisson of the local unfolded statistics
(see~\cite[Theorems 1.3 and 1.6]{Ge-Kl:10}). In the present case,
these results still hold under assumption~\eqref{eq:60}.  Moreover,
the size of intervals over which the uniform convergence of Poisson
statistics is proved in \cite{Ge-Kl:10} can be notably improved thanks
to the improved Wegner and Minami estimates of Theorem~\ref{thmWM} if
the density of states is zero at the point $E_0$. We refer to
Remark~\ref{remsize} for further precisions.  \vskip.1cm
\subsection{Local spectral statistics at spectral edges}
\label{sec:spectr-stat-at}
In dimension one, for any $H_0$ (thus, in particular, for the free
Laplace operator), the condition~\eqref{eq:60} is satisfied at all the
spectral edges, i.e. in the Lifshitz tails region as the Lifshitz
exponent is $1/2$ if the density $g$ does not decay too fast at the
edges of its support (see~\cite{MR1616938}). So, we get the Poisson
behavior for the unfolded eigenvalues at all the spectral edges,
namely,
\begin{Th}
  \label{thr:4}
  Assume $d=1$. Let $E_0\in\partial\Sigma$. \\
  When $|\Lambda|\to+\infty$, the point process
  $\Xi(E_0,\omega,\Lambda)$ converges weakly to a Poisson point
  process on the half line with intensity $1$ if
  $E_0\in\partial\Sigma$, the half-line being $\R^+$ (resp. $\R^-$) if
  $(E_0-\varepsilon,E_0)\cap\Sigma=\emptyset$
  (resp. $(E_0,E_0+\varepsilon)\cap\Sigma=\emptyset$) for some
  $\varepsilon>0$.
\end{Th}
\noindent To the best of our knowledge, this is the first proof of
Poisson asymptotics for the unfolded eigenvalues at the spectral
edges. In~\cite{MR1083950,MR1793333}, for fixed $k$, the authors
studied the joint law of the first $k$ eigenvalues of special
one-dimensional random (continuous) models; for these models, the
density of states $N(E)$ can be computed explicitly. \vskip.1cm
\noindent Theorem~\ref{thr:4} admits an analogue for the continuous
Anderson model in one dimension (with or without a background periodic
potential) provided the single site potential satisfies the
assumptions of~\cite{MR2663411,CGK11}, the study of Lifshitz tails at
band edges being a classical result (see e.g.~\cite{MR1249597}). \\
In higher dimensions, if $H_0$ is the free Laplace operator, the
Lifshitz exponent at spectral edges is usually $d/2$, even more so at
the bottom of the spectrum (see
e.g.~\cite{MR2307751,MR1695202,MR1845184}); so condition~\eqref{eq:60}
is not satisfied in this case. Nevertheless, it may be satisfied if
$H_0$ is not the free Laplace operator as we shall see now.\\
By assumption, the function $h$ defined n~\eqref{eq:19} is real
analytic on $\T^d=\R^d/(2\pi\Z^d)$. Under some additional assumptions
on $h$ near its, say, minimum, one can show that
condition~\eqref{eq:60} is satisfied near the infimum of the almost
sure spectrum of $H_\omega$ (see e.g.~\cite{MR1616938,MR1942859})
\begin{Th}
  \label{thr:5}
  Assume that $\displaystyle\min_{\theta\in\T^d} h(\theta)=0$, that
  $h^{-1}(0)$ is discrete and that, for $\theta_0\in h^{-1}(0)$, there
  exists $\alpha>d/2$ such that, for $\theta$ close to $\theta_0$, one
  has $h(\theta)\leq
  |\theta-\theta_0|^{\alpha}$.\\
  Let $E_-=\inf\Sigma$ where $\Sigma$ is the almost sure spectrum of
  $H_\omega$. \\
  When $|\Lambda|\to+\infty$, the point process
  $\Xi(E_-,\omega,\Lambda)$ converges weakly to a Poisson point
  process on the half line $\R^+$ with intensity $1$.
\end{Th}
\noindent In the case of continuous models, the above analysis could also be useful
in the case of internal edges (see e.g.~\cite{MR1979772}). It also
presumably should be sufficient to deal with the Anderson model in a
constant magnetic field (see e.g.~\cite{MR2249786,MR2655149}).
However we note that for continuous models, the proofs of Minami's estimate of \cite{MR2663411} and \cite{CGK11} do not readily extend to the gap situation.
 \vskip.2cm
\subsection{Ergodicity of the local eigenvalue distribution}
\label{sec:ergod-local-eigenv}
The local results of~\cite{Kl:10a} can also be improved along the same
lines. For $J=[a,b]$ a compact interval such that $N(b)-N(a)=N(J)>0$
and a fixed configuration $\omega$, consider the point process
\begin{equation}
  \label{eq:6}
  \Xi_J(\omega,t,\Lambda) = 
  \sum_{E_n(\omega,\Lambda)\in J}
  \delta_{N(J)|\Lambda|[N_J(E_n(\omega,\Lambda))-t]}
\end{equation}
under the uniform distribution in $[0,1]$ in $t$; here we have set
\begin{equation}
  \label{eq:7}
  N_J(\cdot):=\frac{N(\cdot)-N(a)}{N(b)-N(a)}.  
\end{equation}
\begin{Th}
  \label{thr:2}
  Pick $E_0\in\Sigma_{SDL}$.\\
  Fix $(I_\Lambda)_\Lambda$ a decreasing sequence of intervals such
  that $\displaystyle\sup_{I_\Lambda}|x|\vers_{|\Lambda|\to+\infty}0$.
  Assume that, for some $\delta\in(0,1)$, one has
  \begin{gather}
    \label{eq:63}
    |\Lambda|^{\delta}\cdot N(E_0+I_\Lambda)\to+\infty ,\\
    \intertext{ and }
    \label{eq:64}
    \text{if }\ell'=o(L)\text{ then
    }\frac{N(E_0+I_{\Lambda_{L+\ell'}})}{N(E_0+I_{\Lambda_L})}
    \vers_{|\Lambda|\to+\infty}1.
  \end{gather}
  Then, $\omega$-almost surely, the probability law of the point
  process $\Xi_{E_0+I_\Lambda}(\omega,\cdot,\Lambda)$ under the
  uniform distribution $\car_{[0,1]}(t)dt$ converges to the law of the
  Poisson point process on the real line with intensity $1$.
\end{Th}
\noindent The main improvement over~\cite[Theorem 1.5]{Kl:10a} is that
there is no restriction anymore on the relative sizes of
$N(E_0+I_\Lambda)$ and $|I_\Lambda|$ respectively the density of
states measure and the length of $E_0+I_\Lambda$.
\subsection{Eigenvalue spacings statistics}
\label{sec:eigenv-spacing}

As a consequence of Theorem~\ref{thr:2}, using the results
of~\cite{MR2352280}, we obtain the following result which improves
upon~\cite[Theorem 1.5]{Ge-Kl:10} and~\cite[Theorem 1.5]{Kl:10a} in
the sense that we cover a larger region of energies and the required
lower bound on the IDS is relaxed.
\begin{Th}
  \label{thr:1}
  Fix $E_0\in \Sigma_{SDL}$ and $(I_\Lambda)_\Lambda$ a decreasing
  sequence of intervals satisfying~\eqref{eq:63} and~\eqref{eq:64}.\\
  Define
  \begin{equation}
    \label{eq:3}
    \delta N_j(\omega,\Lambda)=|\Lambda|(N(E_{j+1}(\omega,\Lambda))-
    N(E_j(\omega,\Lambda)))\geq0.
  \end{equation}
  Define the empirical distribution of these spacings to be the random
  numbers, for $x\geq0$
  \begin{equation}
    \label{eq:4}
    DLS(x;E_0+I_\Lambda,\omega,\Lambda)=\frac{\#\{j;\
      E_j(\omega,\Lambda)\in E_0+I_\Lambda,\ \delta
      N_j(\omega,\Lambda)\geq x\}
    }{N(E_0+I_\Lambda,\Lambda,\omega)} ,
  \end{equation}
  where $N(E_0+I_\Lambda,\Lambda,\omega)$ is the random number of
  eigenvalues of $H_\omega(\Lambda)$ in $E_0+I_\Lambda$.\\
  Then, with probability $1$, as $|\Lambda|\to+\infty$,
  $DLS(x;E_0+I_\Lambda,\omega,\Lambda)$ converges uniformly to the
  distribution $x\mapsto e^{-x}$, that is, with probability $1$,
  \begin{equation}
    \label{eq:5}
    \sup_{x\geq0}\left|DLS(x;E_0+I_\Lambda,\omega,\Lambda)
      -e^{-x}\right|\vers_{|\Lambda|\to+\infty}0.
  \end{equation}
\end{Th}
\noindent Spacings statistics over intervals of macroscopic size are
also available (see~\cite[Theorem 1.6]{Ge-Kl:10} and~\cite[Theorem
1.2]{Kl:10a}) in the present context.
\subsection{A large deviation and a central limit theorem for the
  eigenvalue counting function}
\label{sec:large-deviation-an}
Finally, in some regimes, we also can improve upon the large deviation
estimate obtained for the eigenvalue counting function
in~\cite[Theorem 1.9]{Ge-Kl:10} for which we also prove a central
limit theorem in
\begin{Th}
  \label{ThLD}
  For $L>1$, let $\Lambda=\Lambda_L$. Pick a sequence of compact
  intervals $I_\Lambda\subset \Sigma_{SDL}$ so that, for some $1\le \beta
  \le \beta'<\alpha'\le\alpha<\infty$, for all $L$, one has
  \begin{equation}
    \label{eq:17}
    |I_\Lambda|^{-\alpha'}\lesssim|\Lambda| \lesssim
    |I_\Lambda|^{-\alpha}\quad\text{and}
    \quad      |I_\Lambda|^{\beta'}\lesssim N(I_\Lambda)\lesssim
    |I_\Lambda|^{\beta}.
  \end{equation}
  Set
  $\displaystyle\nu_0=\frac1{\alpha-\beta}\min\left(\alpha'-\beta',
    \frac1{d+1}\right)$.
  \begin{enumerate}
  \item {\bf Large deviation estimate.} For $\eps>0$ small enough
    (depending on $\nu_0$), we have
    \begin{multline*}
      \P\left\{ \left| \tr \car_{I_\Lambda}(H_\omega(\Lambda))-
          N(I_\Lambda)|\Lambda| \right| \ge
        (N(I_\Lambda)|\Lambda|)^{\max(\frac12, 1-\nu_0)+\eps} \right\}
      \\\le \exp\left(- (N(I_\Lambda) |\Lambda| )^\eps\right).
    \end{multline*}
  \item {\bf Central limit theorem.} Assume $\nu_0>\frac12$. Then the
    random variable
    \begin{equation*}
      \frac{
        \tr \car_{I_\Lambda}(H_\omega(\Lambda))- N(I_\Lambda)|\Lambda|
      }{(N(I_\Lambda)|\Lambda|)^{\frac12} }   
    \end{equation*}
    converges in law to the standard Normal distribution.
  \end{enumerate}
\end{Th}
\noindent We first point out that only the size of the intervals
$I_\Lambda$ matters, more precisely, their relative size compared to
the volume and the density of states (see \eqref{eq:17}). In
particular, the intervals $I_\Lambda$ need not be centered at a given
point. \\
Let us also note that, by the standard Wegner estimate (W), one can
always pick $\beta=1$.
\vskip.2cm\noindent Let us close this section with a brief outline of
the paper. In Section~\ref{sec:WMapprox}, we prove the enhanced Wegner
and Minami estimates, namely Theorem~\ref{thmWM}. Then, we turn to the
proofs of the results on spectral statistics. In
Section~\ref{sec:reduc}, we prove
\begin{itemize}
\item three different theorems of approximation of the eigenvalues of
  $H_\omega(\Lambda)$ by eigenvalues on smaller cubes as
  in~\cite[Theorems 1.15 and 1.16]{Ge-Kl:10} (each theorem being
  optimized according to a given point of view),
\item the distribution functions of these approximated eigenvalues as
  in ~\cite[Lemma 2.1]{Ge-Kl:10}.
\end{itemize}
In Section~\ref{sec:applic}, we derive the spectral statistics theorem
{\it per se}, namely Theorem~\ref{thr:3}, Theorem~\ref{thr:2},
Theorem~\ref{thr:1} and Theorem~\ref{ThLD} .
\section{The proofs of the enhanced Wegner and Minami estimates}
\label{sec:WMapprox}
We start with the proof of Theorem~\ref{thmWM}. Then, we use it to
derive the distribution of the ``unique'' eigenvalue of
$H_\omega(\Lambda)$ in $I$ when $N(I)|\Lambda|$ is small.
\subsection{Proof of the improved Wegner and Minami estimates}
\label{sec:proof-impr-wegn}
\begin{proof}[Proof of Theorem~\ref{thmWM}]
  The proof of point~(1) is analogous to that of~\cite[Lemma
  2.2]{Ge-Kl:10}. The main gain is obtained by the use of covariance
  which relies on the specific approximations we take for the finite
  volume Hamiltonians.\\
  To that end, note that by covariance
  \begin{equation}
    \label{eq:53}
    \E[\tr\car_{I_\Lambda}(H_\omega(\Lambda))] = |\Lambda|\,
    \E[\tr(\chi_0\car_{I_\Lambda}(H_\omega(\Lambda))\chi_0)].
  \end{equation}
  Recall \eqref{defIDS}: the increase of the integrated density of
  states of $H_\omega$ on $I_\Lambda$ is given by
  \begin{equation*}
    N(I_\Lambda) =  N(b_\Lambda)-N(a_\Lambda)=\E[\tr(\chi_0
    \car_{I_\Lambda}(H_\omega)\chi_0)].
  \end{equation*}
  To control
  $\left|\E[\tr(\chi_0\car_{I_\Lambda}(H_\omega(\Lambda))\chi_0)] -
    \E[\tr(\chi_0\car_{I_\Lambda}(H_\omega)\chi_0)] \right|$, we use
  localization estimates after having smoothed out the characteristic
  function. Let $f_\delta$ be a $\mathcal{C}^\infty$ and compactly
  supported function such that $f_\delta=1$ in $I_\Lambda$, and
  $f_\delta=0$ outside a neighborhood of length $\delta$ of
  $I_\Lambda$ ($\delta$ is small enough so that the support of
  $f_\delta$ lies in $\Sigma_{SDL}$). Note that, by Wegner's estimate (W)
  and the Lipschitz continuity of $N$,
  \begin{equation}
    \label{eq:50}
    \left|\E[\tr(\chi_0[\car_{I_\Lambda}(H_\omega(\Lambda))-
      f_\delta(H_\omega(\Lambda))]\chi_0)]\right|+
    \left|\E[\tr(\chi_0(\car_{I_\Lambda}(H_\omega)-
      f_\delta(H_\omega))\chi_0)] \right|\leq C \delta.
  \end{equation}
  To estimate $\left|\E[\tr(\chi_0(f_\delta(H_\omega(\Lambda))-
    f_\delta(H_{\omega}))\chi_0)]\right|$, we use a Helffer-Sj{\"o}strand
  formula to represent $f_\delta(H_\omega(\Lambda))$ and
  $f_\delta(H_{\omega})$. As the support of $f_\delta$ lies in the
  localization region and as $f$ is of $\mathcal{C}^\infty$
  regularity, the exponential decay estimate for the resolvents of
  $H_\omega(\Lambda)$ and $H_\omega$ imply that, for $\xi'\in(\xi,1)$,
  there exists $C>0$ such that, for $\Lambda$ sufficiently large (see
  e.g. the computation of \cite[(6.23)]{MR2663411}), one has,
  \begin{equation*}
    \left|\E[\tr(\chi_0(f_\delta(H_\omega(\Lambda))-
      f_\delta(H_{\omega}))\chi_0)]\right|\leq\delta^{-C}e^{-\ell^{\xi'}}.
  \end{equation*}
  We set $\delta = e^{-\ell^{\xi''}}$ with $\xi''\in(\xi,\xi')$. Plugging this
  into~\eqref{eq:50} and~\eqref{eq:53} yields~\eqref{Wcontrol} and
  completes the proof of~\eqref{Wcontrol}.
  We turn to point~(2) and shall take advantage of the strategy
  introduced in \cite{MR2663411} to prove a Minami estimate from
  spectral averaging and Wegner estimate. We adapt
  \cite[Theorem~2.3]{MR2505733} to get Minami's estimate in its
  generalized form. The proof is the same except at the very last
  step,~\cite[(4.17)]{MR2505733}, where the estimate~\eqref{W} is used
  in replacement of the usual Wegner bound.
\end{proof}
\begin{Rem}
  \label{rem:2}
  The periodic boundary conditions imposed upon the restrictions
  $H_\omega(\Lambda)$ are important as they enable us to preserve
  covariance. \\
  This covariance for the restricted operators can also be achieved
  for some continuous models e.g. for models having the following
  structure
  \begin{equation}
    \label{eq:16}
    H_\omega=H_0+\sum_{\gamma\in\Z^d}V(x-\gamma,\omega_\gamma)
  \end{equation}
  where $(\omega_\gamma)_\gamma$ are i.i.d. random variables and $H_0$
  is e.g. $\Z^d$-periodic.\\
  Then, for $\Lambda=[0,L]^d$, as local restriction
  $H_\omega(\Lambda)$, one considers the restriction of following
  $L\Z^d$-periodic operator to $\Lambda$
  \begin{equation*}
    H^L_\omega=H_0+\sum_{\beta\in\Z^d}\sum_{\gamma\in\Lambda}
    V(x-\gamma-L\beta,\omega_\gamma).
  \end{equation*}
  It is well known that such restrictions yield very fast convergence
  towards the integrated density of states (see
  e.g.~\cite{MR1695202,MR1902460}).\\
  This can be used for the continuous Anderson model considered
  in~\cite{MR2663411}.
\end{Rem}
\subsection{The distribution of the ``local'' eigenvalues}
\label{sec:distr-local-eigenv}
Consider a cube $\Lambda$ of side length $\ell$ and an interval
$I_\Lambda=[a_\Lambda,b_\Lambda]\subset I$ (i.e. $I_\Lambda$ is
contained in the localization region). Consider the following random
variables:
\begin{itemize}
\item $X=X(\Lambda,I_\Lambda)$ is the Bernoulli random variable
  \begin{equation*}
    X=\car_{H_\omega(\Lambda)\text{ has exactly one
        eigenvalue in }I_\Lambda};
  \end{equation*}
\item $\tilde E=\tilde E(\Lambda,I_\Lambda)$ is the eigenvalue of
  $H_\omega(\Lambda)$ in $I_\Lambda$ conditioned on $X=1$;
\item $\tilde\xi=\tilde\xi(\Lambda,I_\Lambda)=(\tilde
  E(\Lambda,I_\Lambda)-a_\Lambda)/|I_\Lambda|$.
\end{itemize}
Clearly $\tilde\xi$ is valued in $[0,1]$; let $\tilde\Xi$ be its
distribution
function.\\
In the present section, we will describe the distribution of these
random variables as $|\Lambda|\to+\infty$ and $|I_\Lambda|\to0$. We
prove
\begin{Le}
  \label{lemasympt}
  For any $\nu\in(0,1)$ and $K$ compact interval in $\Sigma_{SDL}$ ,
  there exists $C>1$ such that, for $\Lambda=\Lambda_\ell$ and
  $I_\Lambda\subset K$ such that $N(I_\Lambda)\geq e^{-\ell^{\nu}/C}$,
  one has
  \begin{equation}
    \label{eq:51}
    \pro(X=1)=N(I_\Lambda)|\Lambda|\left(1+O(|I_\Lambda||\Lambda|)+ O
      \left(e^{-\ell^\nu}\right)\right)
  \end{equation}
  where $O(\cdot)$ are locally uniform in $\Sigma_{SDL}$.\\
  Moreover, for $(x,y)\in[0,1]$, one has
  \begin{multline}
    \label{eq:52}
    (\tilde\Xi(x)-\tilde\Xi(y))\,P(X=1)=\\\left[N(a_\Lambda+x|I_\Lambda|)-
      N(a_\Lambda+y|I_\Lambda|)\right]|\,\Lambda|\,(1+O(|x-y||I_\Lambda||\Lambda|)\\+
    O\left((|x-y|)^{-C}e^{-\ell^\nu/C}\right)
  \end{multline}
  where $O(\cdot)$ are locally uniform in $\Sigma_{SDL}$.
\end{Le}
\noindent This lemma is to be compared with~\cite[Lemma
2.1]{Ge-Kl:10}; it gives a fairly good description of the random
variables $X$ and $\tilde\xi$ if $|I_\Lambda||\Lambda|\ll1$.
\begin{proof}[Proof of Lemma~\ref{lemasympt}]
  We follow the proof of~\cite[Lemma 2.1]{Ge-Kl:10}. Using~\eqref{M},
  the estimate~\cite[(2.5)]{Ge-Kl:10} becomes
  \begin{equation*}
    0\leq \E(\tr[\car_{I_\Lambda}(H_\omega(\Lambda))])-\P(X=1)
    \leq C N(I_\Lambda)|I_\Lambda||\Lambda|^2
  \end{equation*}
  Thus, the proof of~\cite[Lemma 2.1]{Ge-Kl:10} yields
  \begin{equation}
    \label{eq:8}
    \pro(X=1)=N(I_\Lambda)|\Lambda|(1+O(|I_\Lambda||\Lambda|)+ O
    \left(|I_\Lambda|^{-C}e^{-\ell^\nu/C}\right)
  \end{equation}
  Recall that, by Wegner's estimate, $N(I_\Lambda)|\Lambda|\leq
  C|I_\Lambda||\Lambda|$. Thus, as $N(I_\Lambda)\geq
  e^{-\ell^{\nu}/C}$, enlarging possibly $C$, one
  obtains~\eqref{eq:51}.\\
  Replacing $I_\Lambda$ with the interval $I_{x,y,\Lambda}$ in the
  estimation of $P(X=1)$ yields the proof of~\eqref{eq:52}.\\
  The proof of Lemma~\ref{lemasympt} is complete.
\end{proof}

\section{Box reduction and description of the eigenvalues}
\label{sec:reduc}
After giving a precise description of $\Sigma_{SDL}$, we shall state
three different reductions, that improve on the ones given in
\cite{Ge-Kl:10} thanks to Theorem~\ref{thmWM}. Each of them is
optimized according to a given point of view, depending on the
application it will be used for (Poisson statistics, local ergodicity
of eigenvalues and level spacings statistics, or deviation estimates
and CLT).
\subsection{The strong dynamical localization regime}
\label{sec:strong-dynam-local}
Before turning to the local description of the eigenvalues, let us
underline that we will give this description in the discrete
setting. In particular, we will use the fact that in the localization
region $ \Sigma_{SDL}$, we have true exponential decay. In the
continuous setting, one generally only has sub-exponential decay (see
e.g.~\cite{MR2002h:82051,MR2078370,Ge-Kl:10}). The only change that
this implies is that the length scales of size $\log|\Lambda|$ have to
be replaced by length scales of size $\log^{1/\xi}|\Lambda|$ (for some
$\xi\in(0,1)$).\\
We first recall
\begin{Th}[\cite{Ge-Kl:10}]
  \label{thmFVE}
  Let $I\subset\Sigma$ be a compact interval and assume that Wegner's
  estimate {\bf (W)} holds in $I$. For $L$ given, consider $\Lambda=
  \Lambda_L(0)$ a cube of side length $L$ centered at $0$, and denote
  by $\varphi_{\omega,\Lambda,j}$, $j=1,\cdots,
  \mathrm{tr}\car_I(H_\omega(\Lambda))$, the normalized eigenvectors
  of $H_\omega(\Lambda)$ with corresponding eigenvalue in $I$. The
  following are equivalent
  \begin{enumerate}
  \item $I\subset \Sigma_{\mathrm{SDL}}$
  \item For all $E\in I$, there exists $\theta>3d-1$,
    \begin{equation}
      \label{startMSA}
      \limsup_{L\to\infty} \P\left\{\forall x,y\in\Lambda, \,
        |x-y|\ge \frac L2, \;  \| \chi_x (H_\omega(\Lambda) - E)^{-1}
        \chi_y  \|  \le L^{-\theta} \right\} =1. 
    \end{equation}
  \item There exists $\xi\in(0,1)$,
    \begin{equation}
      \label{FVsudecE}
      \sup_{y\in \Lambda} \E\left\{ \sum_{x\in\Lambda} \e^{\xi|x-y|}
        \sup_j \|  \varphi_{\omega,\Lambda,j} \|_x \|
        \varphi_{\omega,\Lambda,j} \|_y  \right\}<\infty. 
    \end{equation}
  \item There exists $\xi\in(0,1)$,
    \begin{equation}
      \label{FVsudecE2}
      \sup_{y\in \Lambda} \E\left\{ \sum_{x\in\Lambda}  \e^{\xi|x-y|}
        \sup_j \|  \varphi_{\omega,\Lambda,j} \|_x \|
        \varphi_{\omega,\Lambda,j} \|_y  \right\}<\infty. 
    \end{equation}
  \item There exists $\xi\in(0,1)$,
    \begin{equation}
      \label{FVdynloc1}
      \sup_{y\in \Lambda} \E\left\{ \sum_{x\in\Lambda}  \e^{\xi|x-y|}
        \sup_{\substack{\supp f \subset I \\ |f|\le 1 }} \| \chi_x
        f(H_\omega(\Lambda)) \chi_y \|_2  \right\}<\infty. 
    \end{equation}
  \item There exists $\xi\in(0,1)$,
    \begin{equation}
      \label{FVdynloc2}
      \sup_{y\in \Lambda} \E\left\{ \sum_{x\in\Lambda}  \e^{\xi|x-y|}
        \sup_{\substack{\supp f \subset I \\ |f|\le 1 }} \| \chi_x
        f(H_\omega(\Lambda)) \chi_y \|_2  \right\}<\infty.
    \end{equation}
  \item There exists $\xi\in(0,1)$,
    \begin{equation}
      \label{FVdynloc3}
      \sup_{y\in \Lambda} \sup_{\substack{\supp f \subset I \\ |f|\le 1 }}
      \E\left\{ \sum_{x\in\Lambda}  \e^{\xi|x-y|} \| \chi_x
        f(H_\omega(\Lambda)) \chi_y \|_2  \right\}<\infty. 
    \end{equation}
  \item (SUDEC for finite volume with polynomial probability) For all
    $p>d$, there is $q=q_{p,d}$ so that for some $\xi\in(0,1)$, for any
    $L$ large enough, the following holds with probability at least $1
    - L^{-p}$: for any eigenvector $\varphi_{\omega,\Lambda,j}$ of
    $H_{\omega,\Lambda}$, with energy in $I$, for any
    $(x,y)\in\Lambda^2$, one has
    \begin{equation}
      \label{FVsudecP}
      \|  \varphi_{\omega,\Lambda,j} \|_x \|
      \varphi_{\omega,\Lambda,j} \|_y\le L^{q} \e^{-\xi|x-y|}. 
    \end{equation}
  \item (SULE for finite volume with polynomial probability) For all
    $p>d$, there is $q=q_{p,d}$ so that, for some $\xi\in(0,1)$, for
    any $L$ large enough, the following holds with probability at
    least $1 - L^{-p}$: for any eigenvector
    $\varphi_{\omega,\Lambda,j}$ of $H_{\omega,\Lambda}$, with energy
    in $I$, there is a center of localization $x_{\omega,\Lambda,j}\in
    \Lambda$, so that for any $x\in \Lambda$, one has
    \begin{equation}
      \label{FVsule}
      \|  \varphi_{\omega,\Lambda,j} \|_x \le L^q
      \e^{-\xi|x-x_{\omega,\Lambda,j}|}. 
    \end{equation}
  \item (SUDEC for finite volume with sub-exponential probability) For
    all $\nu,\xi\in(0,1)$, $\nu<\xi$, for any $L$ large enough, the
    following holds with probability at least $1 - \e^{-L^\nu}$: for
    any eigenvector $\varphi_{\omega,\Lambda,j}$ of
    $H_{\omega,\Lambda}$, with energy in $I$, for any
    $(x,y)\in\Lambda^2$, one has
    \begin{equation}
      \label{FVsudecP1}
      \|  \varphi_{\omega,\Lambda,j} \|_x \|
      \varphi_{\omega,\Lambda,j} \|_y \le \e^{2L^\nu} \e^{-\xi|x-y|}. 
    \end{equation}
  \item (SULE for finite volume with sub-exponential probability)For
    all $\nu,\xi\in(0,1)$, $\nu<\xi$, for any $L$ large enough, the
    following holds with probability at least $1 - \e^{-L^\nu}$: for
    any eigenvector $\varphi_{\omega,\Lambda,j}$ of
    $H_{\omega,\Lambda}$, with energy in $I$, there is a center of
    localization $x_{\omega,\Lambda,j}\in \Lambda$, so that for any
    $x\in \Lambda$, one has
    \begin{equation}
      \label{FVsule1}
      \|  \varphi_{\omega,\Lambda,j} \|_x \le  \e^{2L^\nu} 
      \e^{-\xi|x-x_{\omega,\Lambda,j}|}. 
    \end{equation}
  \end{enumerate}
  Moreover one can pick $q=p+d$ in (8) and $q=p+\frac 32 d$ in (9).
\end{Th}
\subsection{Controlling all the eigenvalues}
\label{sec:contr-all-eigenv}
Assume $E_0$ is such that~\eqref{eq:60} holds for some
$\rho\in(0,1/d)$. Let us first explain why the restriction $\rho<1/d$
is necessary. For an interval $I_\Lambda$ to contain a large number of
eigenvalues of $H_\omega(\Lambda)$ (at least in expectation), one
needs that $N(I_\Lambda)|\Lambda|$ be large. On the other hand, as we
shall see in the proof of the next result, we also need that
$|I_\Lambda|(\log|\Lambda|)^d$ be small. This second restriction is
essentially enforced by the localization of the eigenfunctions in
region of (linear) size $\log|\Lambda|$. These two requirements can
only be met if~\eqref{eq:60} holds.\\
We prove
\begin{Th}
  \label{thr:vsmall1}
  Assume $E_0$ is such that~\eqref{eq:60} holds for some
  $\rho\in(0,1/d)$. Fix $0<d\rho<\rho'<\rho''<1$ and
  $0<\alpha<\min(d/\rho'-d/\rho'',1/\rho-d/\rho')$. Pick $I_\Lambda$
  centered at $E_0$ such that
  $N(I_\Lambda)|\Lambda|=\log^\alpha|\Lambda|$. For $L$ sufficiently
  large and $\ell=\log^{1/\rho'}L$, we have a decomposition of
  $\Lambda$ into disjoint cubes of the form
  $\Lambda_\ell(\gamma_j):=\gamma_j+[0,\ell]^d$:
  \begin{itemize}
  \item $\cup_j\Lambda_\ell(\gamma_j)\subset\Lambda_L$,
  \item $\dist (\Lambda_\ell(\gamma_j),\Lambda_\ell(\gamma_k))\ge
    \log^{1/\rho''}|\Lambda|$ if $j\not=k$,
  \item $\dist (\Lambda_\ell(\gamma_j),\partial\Lambda)\ge
    \log^{1/\rho''}|\Lambda|$
  \item $|\Lambda_L\setminus\cup_j\Lambda_\ell(\gamma_j)|\lesssim |
    \Lambda|\log^{d/\rho''-d/\rho'}|\Lambda|$,
  \end{itemize}
  such that, there exists a set of configurations $\mathcal{Z}_
  \Lambda $ s.t.:
  \begin{itemize}
  \item $\P(\mathcal{Z}_\Lambda)\geq 1 - (\log
    L)^{-(\min(d/\rho'-d/\rho'',1/\rho-d/\rho')-\alpha)/2}$,
  \item for $\omega\in\mathcal{Z}_\Lambda $, each centers of
    localization associated to $H_\omega(\Lambda)$ belong to some
    $\Lambda_\ell(\gamma_j)$ and each box $\Lambda_\ell(\gamma_j)$
    satisfies:
    \begin{enumerate}
    \item the Hamiltonian $H_\omega(\Lambda_\ell(\gamma_j))$ has at
      most one eigenvalue in $I_\Lambda $, say,
      $E(\omega,\Lambda_\ell(\gamma_j))$;
    \item $\Lambda_\ell(\gamma_j)$ contains at most one center of
      localization, say $x_{k_j}(\omega,L)$, of an eigenvalue of
      $H_\omega(\Lambda)$ in $I_\Lambda $, say $E_{k_j}(\omega,L)$;
    \item $\Lambda_\ell(\gamma_j)$ contains a center
      $x_{k_j}(\omega,L)$ if and only if
      $\sigma(H_\omega(\Lambda_\ell(\gamma_j)))\cap
      I_\Lambda\not=\emptyset$; in which case, one has, with $\ell'=
      \log^{-1/\rho''} L$,
      \begin{equation}
        \label{error}
        |E_{k_j}(\omega,L)-E(\omega,\Lambda_\ell(\gamma_j))| \leq  
        e^{-\ell'}\text{ and }\mathrm{dist}(x_{k_j}(\omega,L),
        \Lambda_L \setminus \Lambda_\ell(\gamma_j))\geq \ell'.
      \end{equation}
    \end{enumerate}
  \end{itemize}
  In particular, if $\omega\in\mathcal{Z}_\Lambda$, all the
  eigenvalues of $H_\omega(\Lambda)$ are described by~\eqref{error}.
\end{Th}
\begin{Rem}
  \label{remsize}
(i) In Theorem~\ref{thr:vsmall1}, the condition
$N(I_\Lambda)|\Lambda|=\log^\alpha|\Lambda|$ does not, in general,
provide the largest possible interval $I_\Lambda$ where our analysis
works. It is chosen so as to work in all regimes
provided~\eqref{eq:60} holds; it is optimal only in regimes where the
integrated density of states $N(I_\Lambda)$ is exponentially small in
$|I_\Lambda|^{-1}$. In other regimes, one may actually take
$N(I_\Lambda)|\Lambda|$ larger.  Note that, as in
Theorem~\ref{thr:vsmall1}, one has $|I_\Lambda|\ell^d\ll1$,
Lemma~\ref{lemasympt} gives a precise description of:
\begin{itemize}
\item the probability distribution of the $\gamma$'s for which
  $H_\omega(\Lambda_\ell(\gamma))$ has exactly one eigenvalue in
  $I_\Lambda$,
\item the distribution of this eigenvalue when this is the case
\end{itemize}

 (ii) If the integrated density of states $N(I_\Lambda)$ is exponentially
  small in $|I_\Lambda|^{-1}$ then as pointed out in \eqref{sizeI}
  below, the typical size of intervals where we can control all the
  eigenvalues, and thus prove Poisson convergence, is of order an
  inverse power of $\log |\Lambda|$. This should be compared to
  \cite{Ge-Kl:10} where the admissible size for $|I_\Lambda|$ was of
  order $|\Lambda|^{-\alpha}$, with $\alpha > (1+(d+1)^{-1})^{-1}$
  (see \cite[(1.43)]{Ge-Kl:10} with $\rho=1$ and $\rho'=0$). If now we have
  $N(I_\Lambda)\asymp |I_\Lambda|^{1+\rho'}$, $\rho'\ge 0$, then the
  admissible size for $|I_\Lambda|$ is of order $|\Lambda|^{-\alpha}$,
  with $\alpha > (1+\rho'+(d+1)^{-1})^{-1}$, with no restriction on
  $\rho'$ (compare to \cite[(1.12)]{Ge-Kl:10}). In particular $\alpha$
  can be close to zero if $\rho'$ is large. These improvements are
  direct consequences of the improved Wegner and Minami estimates of
  Theorem~\ref{thmWM} above.
\end{Rem}
\begin{proof}[Proof of Theorem~\ref{thr:vsmall1}]
  As $N(I_\Lambda)|\Lambda|=\log^\alpha|\Lambda|$,
  assumption~\eqref{eq:60} yields, for $|\Lambda|$ sufficiently large,
  \begin{equation}
    |I_\Lambda|\leq 2\log^{-1/\rho}|\Lambda|.  \label{sizeI}
  \end{equation}
  We follow the proof of~\cite[Theorem 1.15]{Ge-Kl:10}. First we note
  that, by the localization property (Loc), we need $\ell$ to be
  larger than $\log^{1/\rho'}L$ for some large $C$ to
  get~\eqref{error}. With the choices made in
  Theorem~\ref{thr:vsmall1}, the estimate (3.1) in~\cite{Ge-Kl:10}
  becomes
  \begin{equation}
    \label{eq:1}
    \begin{split}
      \P(\# \mathcal{S}_{\ell,L} \ge 1) &\lesssim
      \frac{L^d}{\log^{d/\rho'} L} N(I_\Lambda)|I_\Lambda|(\log
      L)^{2d/\rho'} \lesssim L^d\,\left(\log^{d/\rho'}L\right)\,
      N(I_\Lambda)|I_\Lambda|\\&\lesssim
      \log^{\alpha-(1/\rho-d/\rho')}L
    \end{split}
  \end{equation}
  by our choice of $I_\Lambda$.\\
  In the same way, the estimate (3.3) in~\cite{Ge-Kl:10} becomes
  \begin{equation}
    \label{eq:2}
    \begin{split}
      & \hskip-1cm \P(H_\omega(\Lambda) \mbox{ has a localization
        center in } \Upsilon) \\ &\lesssim |\Upsilon|
      N(I_\Lambda)\\&\lesssim |I_\Lambda|^{-1} N(I_\Lambda)|I_\Lambda|
      |\Lambda|\left(\log^{d/\rho''-d/\rho'}|\Lambda|\right)
      \\&\lesssim \log^{\alpha-d(1/\rho'-1/\rho'')}|\Lambda|
    \end{split}
  \end{equation}
  by our choice of $I_\Lambda$.\\
  This completes the proof of Theorem~\ref{thr:vsmall1}.
\end{proof}
\subsection{Controlling most eigenvalues}
\label{sec:contr-most-eigenv}
We will give two versions of this reduction. In
Theorem~\ref{thr:vbig1}, the first version, we consider region where
the density of states is not too small: it can be polynomially small
to any order but not smaller. In this region, we give a version of the
reduction that minimizes the estimate on the probability of the bad
set (where our description does not work) as well as the number of
eigenvalues that is not described by our scheme.\\
In Theorem~\ref{thr:vbig0}, the second version of the reduction
theorem, we want to allow exponentially small density of states as in
Theorem~\ref{thr:vsmall1}. The control will still be obtained with a
good probability but we don't control as many eigenvalues. This
version is used in the proofs of Theorems~\ref{thr:2} and~\ref{thr:1}.\\
The reduction goes back to the results obtained in
\cite[Theorem~1.14]{Ge-Kl:10} and improves upon it. We follow the
proof of that result and only indicate the differences.
\begin{Th}
  \label{thr:vbig1}
  Set $\ell'= R\log |\Lambda|$ with $R$ large and consider intervals
  $I_\Lambda\subset \Sigma_{SDL}$. Assume that for some $1\le \beta \le
  \beta'<\alpha'\le\alpha<\infty$, for $|\Lambda|$ large, we have
  \begin{equation*}
    |I_\Lambda|^{-\alpha'}\lesssim |\Lambda|\lesssim
    |I_\Lambda|^{-\alpha}\quad\text{ and }\quad |I_\Lambda|^{\beta'}
    \lesssim N(I_\Lambda)\lesssim |I_\Lambda|^{\beta} 
  \end{equation*}
  Set
  \begin{align}\label{set}
    \delta_0 = (\alpha-\beta)^{-1}>0, \quad \zeta= \frac{
      \alpha-\beta}{ \alpha'-\beta'}\ge 1, \quad \nu_0=
    \min\left(\zeta^{-1}, \frac {\delta_0}{d+1}\right)\leq1.
  \end{align}
  Note that $\nu_0\zeta\le 1$.\\
  For any $\nu<\nu_0$ and $\kappa\in(0,1)$, there exists
  \begin{itemize}
  \item a decomposition of $\Lambda_L$ into
    $\mathcal{O}(|\Lambda|/\ell_\Lambda^d)$ disjoint cubes of the form
    $\Lambda_\ell(\gamma_j):=\gamma_j+[0,\ell]^d$, where $\ell\sim
    (|I_\Lambda| \ell')^{-\frac 1{d+1}}$, so that:
    \begin{itemize}
    \item $\cup_j\Lambda_\ell(\gamma_j)\subset\Lambda_L$,
    \item $\dist (\Lambda_\ell(\gamma_j),\Lambda_\ell(\gamma_k))\ge
      \ell'$ if $j\not=k$,
    \item $\dist (\Lambda_\ell(\gamma_j),\partial\Lambda)\ge \ell'$
    \item $|\Lambda_L\setminus\cup_j\Lambda_\ell(\gamma_j)|\lesssim |
      \Lambda| \ell'/\ell$,
    \end{itemize}
  \item a set of configurations $\mathcal{Z}_\Lambda$ satisfying
    $\displaystyle \P(\mathcal{Z}_\Lambda)\geq 1 - \exp\left(-
      (N(I_\Lambda) |\Lambda|)^{\delta_{\kappa,\nu}}/C\right)$, with
    $\delta_{\kappa,\nu}=\min(1- \nu\zeta,\kappa\nu)\in]0,\frac12[$,
  \end{itemize}
  such that, for $|\Lambda|$ sufficiently large (depending only on
  $\beta$, $\beta'$, $\alpha$, $\alpha'$, $\nu$ and $\kappa$),
  \begin{itemize}
  \item for all $\omega\in\mathcal{Z}_\Lambda$, there exists at least
    $\displaystyle \frac{|\Lambda|}{\ell^d}(1+o\left(1\right))$
    disjoint boxes $\Lambda_\ell(\gamma_j)$ satisfying the properties
    (1), (2) and (3) described in Theorem~\ref{thr:vsmall1} with
    $\ell'= R\log |\Lambda|$,
  \item the number of eigenvalues of $H_{\omega,L}$ that are not
    described by the above picture is bounded by $C
    (N(I_\Lambda)|\Lambda|)^{\gamma_{\kappa,\nu}}$, with
    $\gamma_{\kappa,\nu}=1-(1-\kappa)\nu\in]0,1[$.
  \end{itemize}
  Particular cases:
  \begin{itemize}
  \item If $|I_\Lambda|\asymp |\Lambda|^{-\alpha^{-1}}$ and
    $N(I_\Lambda)\asymp |I_\Lambda|^{\beta}$, then $\zeta=1$.
  \item If $n(E)>0$ and $I_\Lambda$'s are centered at $E$, then
    $\beta=\beta'=1$.
  \end{itemize}
\end{Th}
\begin{proof}[Proof of Theorem~\ref{thr:vbig1}]
  We proceed as in \cite[Proof of Theorem~1.14]{Ge-Kl:10} but take
  advantage of Theorem~\ref{thmWM} above.\\
  By assumption, we have
  \begin{equation}
    |I_\Lambda|^{-(\alpha'-\beta')} \lesssim N(I_\Lambda)|\Lambda|
    \lesssim |I_\Lambda|^{-(\alpha-\beta)} . 
  \end{equation}
  or equivalently, with notations defined in~\eqref{set},
  \begin{equation}
    \label{compare}
    (N(I_\Lambda)|\Lambda|)^{-\zeta\delta_0} \lesssim  |I_\Lambda|
    \lesssim (N(I_\Lambda)|\Lambda|)^{-\delta_0} . 
  \end{equation}
  First, use in the reduction procedure~\eqref{W} and~\eqref{M} as the
  Wegner and Minami estimates. As in \cite{Ge-Kl:10}, we consider a
  collection of $\mathcal{O}(|\Lambda|\ell^{-d})$ boxes
  $\Lambda_\ell(\gamma_j)$ two by two distant by at least $\ell'$, and
  such that $|\Lambda\setminus \cup_j \Lambda_\ell(\gamma_j)|\lesssim
  |\Lambda|\ell'/\ell$. Let $\mathcal{S}_{\ell,L}$ be the set of boxes
  $\Lambda_{\ell}(\gamma_j)\subset\Lambda$ containing at least two
  centers of localization of $H_{\omega,L}$ (see the proof
  of~\cite[Theorem 1.16]{Ge-Kl:10} for more details). Set
  \begin{equation}
    \label{defK}
    K:=4\e  N(I_\Lambda)   |\Lambda| |I_\Lambda| \ell^d
  \end{equation}
  Then, by~\eqref{M},
  \begin{equation}
    \label{probaK}
    \P(\# \mathcal{S}_{\ell,L}  \ge K)   \lesssim 2^{-K}. 
  \end{equation}
  Next, cover $ \Lambda\setminus \cup_j \Lambda_\ell(\gamma_j)$ with a
  partition of boxes $\Lambda_{\ell'}(\gamma'_j)$ and let
  $\mathcal{S}'_{\ell,L}$ be the set of boxes
  $\Lambda_{\ell'}(\gamma_j)\subset\Lambda$ containing at least one
  center of localization of $H_{\omega,L}$ (see the proof
  of~\cite[Theorem 1.16]{Ge-Kl:10} for more details). Set
  \begin{equation}
    \label{defK'}
    K':=2^{d+1}N(I_\Lambda) |\Lambda|  \frac { \ell'}{\ell} .
  \end{equation}
  It follows from~\eqref{W} that
  \begin{equation}
    \label{probaK'}
    \P(\# \mathcal{S}'_{\ell,L}  \ge K')   \lesssim 2^{-K'}. 
  \end{equation}
  To evaluate the number of eigenvalues of $H_\omega(\Lambda)$ we may
  miss because of this reduction, we need to control the number of
  centers $x_k(\omega,\Lambda)$ that may fall into $K$ boxes of
  $\mathcal{S}_{\ell,L}$ and $K'$ boxes of $\mathcal{S}'_{\ell,L}$. In
  \cite{Ge-Kl:10} we used the crude deterministic bound given by the
  volume of the considered boxes. Here we estimate this number using
  the high order Minami estimate~\eqref{HOM}. Given an integer $r\ge
  1$, it follows from~\eqref{HOM} that
  \begin{equation}
    \label{estr2}
    \begin{split}
      \P\left\{
        \begin{aligned}
          \exists \mbox { a box } \Lambda_\ell(\gamma_j) \mbox {
            s.t.}\\ \tr
          \car_{I_\Lambda}(H_\omega(\Lambda_\ell(\gamma_j))) \ge r
        \end{aligned}
      \right\} & \lesssim \frac{|\Lambda|}{\ell^d} \P\{ \tr
      \car_{I_\Lambda}(H_\omega(\Lambda_\ell)) \ge r\}\\
      & \lesssim \frac1{r!}  N(I_\Lambda)|\Lambda| (C |I_\Lambda|
      \ell^d)^{r-1} \\
      & \lesssim N(I_\Lambda)|\Lambda| \left(\frac{C|I_\Lambda|
          \ell^d}r\right)^{r-1}.
    \end{split}
  \end{equation}
  In the same way, we have, for $r'\ge 1$ an integer,
  \begin{equation}
    \label{estr3}
    \begin{split}
      \P\left\{
        \begin{aligned}
          \exists \mbox { a box } \Lambda_{\ell'}(\gamma_j')\mbox {
            s.t.}\\ \tr
          \car_{I_\Lambda}(H_\omega(\Lambda_{\ell'}(\gamma_j'))) \ge
          r'
        \end{aligned}
      \right\} & \lesssim \frac{|\Lambda|}{\ell^d}
      \frac{\ell'}{\ell}\P\{
      \tr \car_{I_\Lambda}(H_\omega(\Lambda_{\ell'})) \geq r'\}  \\
      & \leq N(I_\Lambda)|\Lambda|
      \left(\frac{\ell'}{\ell}\right)^{d+1} \left(\frac{C|I_\Lambda|
          (\ell')^d}{r'}\right)^{r'/C-1}.
    \end{split}
  \end{equation}
  The additional constant $C$ in the exponent in the right hand side
  of~\eqref{estr3} comes form the fact that the Hamiltonians
  $(H_\omega(\Lambda_{\ell'})_{\Lambda_{\ell'}\in
    \mathcal{S}'_{\ell,L}}$ need not be independent,
  but there are finitely many subfamilies of independent Hamiltonians.\\
  We pick $r\asymp r'\asymp (N(I_\Lambda)|\Lambda|)^\delta$, with
  $\delta>0$ to be chosen later. We thus end up with
  \begin{align}
    ~\eqref{estr2}, \,~\eqref{estr3} \lesssim
    \exp(-(N(I_\Lambda)|\Lambda|)^\delta) .
    \label{estr4}
  \end{align}
  Hence, with a probability at least
  $1-\exp(-(N(I_\Lambda)|\Lambda|)^\delta) $, the number of
  eigenvalues of $H_\omega(\Lambda)$ we miss with our reduction is
  bounded by
  \begin{equation}
    \label{bad}
    C (K+K')  (N(I_\Lambda)|\Lambda|)^\delta.
  \end{equation}
  We now optimize in $\ell$ by requiring $K\sim K'$, that is
  \begin{equation}\label{condeq}
    |I_\Lambda| \ell^d  = \frac{\ell'}{\ell}.
  \end{equation}
  In other terms we choose $\ell= (|I_\Lambda| \ell')^{-\frac
    1{d+1}}$. So that any $K,K'$ satisfying
  \begin{equation}
    K\sim K' \gtrsim K_0:=  (N(I_\Lambda) |\Lambda|)
    |I_\Lambda|^{\frac 1{d+1}} ( \ell')^{1+\frac 1{d+1}}, 
  \end{equation}
  is good enough. To ensure that $K$ and $K'$ grow fast enough (in
  order to get a fast decaying probability in~\eqref{probaK}
  and~\eqref{probaK'}), we actually enlarge them and set for any
  $\nu<\nu_0$ (note that $\log(N(I_\Lambda)|\Lambda|)
  \asymp\log|\Lambda|$),
  \begin{equation}
    \label{defKK'}
    K\sim K' \sim (N(I_\Lambda) |\Lambda|)  |I_\Lambda|^{\nu \delta_0^{-1}} .
  \end{equation}
  Taking~\eqref{compare} into account, we have the following lower and
  upper bound,
  \begin{equation}
    \label{KK'bound}
    (N(I_\Lambda)|\Lambda|)^{1-\nu\zeta} \lesssim  K\sim K' \lesssim
    (N(I_\Lambda)|\Lambda|)^{1-\nu}.
  \end{equation}
  Next, let $\kappa\in]0,1[$ be given, and fix $\delta$ above so that
  $\delta=\kappa \nu$. It follows from~\eqref{probaK},
  \eqref{probaK'},~\eqref{estr4},~\eqref{bad} and~\eqref{KK'bound}
  that the number of missing eigenvalues is bounded by $1 -
  \exp(-(N(I_\Lambda)|\Lambda|)^{1-\zeta\nu})
  -\exp(-(N(I_\Lambda)|\Lambda|)^{\kappa \nu})$. At last, to see that
  $\delta_{\kappa,\nu}<\frac 12$, note that
  $\min(1-\zeta\nu,\kappa\nu) \le (1+\zeta\kappa^{-1})^{-1}<\frac 12$,
  since $\kappa<1$ and $\zeta\ge 1$. The theorem follows.
\end{proof}

\noindent We now turn to the second version of our reduction
theorem. One has
\begin{Th}
  \label{thr:vbig0}
  Pick $\rho\in(0,1/d)$. For a cube $\Lambda=\Lambda_L$, consider an
  interval $I_\Lambda=[a_\Lambda,b_\Lambda]\subset I$, $I$ a fixed
  compact in $\Sigma_{SDL}$. Pick  $\rho<\rho'<\rho''<1/d$. \\
  There exists $\alpha_0>0$ such that, for
  $\alpha_0\leq\alpha<\alpha'$, if $(I_\Lambda)$ satisfies
  \begin{equation}
    \label{eq:9}
    \log^\alpha|\Lambda|\leq N(I_\Lambda)|\Lambda|\leq
    \log^{\alpha'}|\Lambda|\quad\text{and}\quad
    N(I_\Lambda)\,e^{|I_\Lambda|^{-\rho}}\geq1.
  \end{equation}
  then, picking $(\tilde\ell_\Lambda,\ell'_\Lambda)$ such that
  \begin{equation}
    \label{eq:24}
    \tilde\ell^d_\Lambda|I_\Lambda|\asymp\log^{1/\rho'-1/\rho}|\Lambda|
    \quad\text{ and
    }\quad
    (\ell'_\Lambda)^d\asymp\tilde\ell^d_\Lambda\log^{1/\rho''-1/\rho'}
    |\Lambda|,
  \end{equation}
  there exists
  \begin{itemize}
  \item a decomposition of $\Lambda_L$ into disjoint cubes of the form
    $\Lambda_{\ell_\Lambda}(\gamma_j):=\gamma_j+[0,\ell_\Lambda]^d$
    where $\displaystyle\ell_\Lambda=\tilde\ell_\Lambda(1+o(1))$ such
    that
    \begin{itemize}
    \item $\cup_j\Lambda_{\ell_\Lambda}(\gamma_j)\subset\Lambda_L$,
    \item $\dist
      (\Lambda_{\ell_\Lambda}(\gamma_j),\Lambda_{\ell_\Lambda}(\gamma_k))\ge
      \ell'_\Lambda$ if $j\not=k$,
    \item $\dist (\Lambda_{\ell_\Lambda}(\gamma_j),\partial\Lambda)\ge
      \ell'_\Lambda$
    \item
      $|\Lambda_L\setminus\cup_j\Lambda_{\ell_\Lambda}(\gamma_j)|\lesssim
      | \Lambda| \ell'_\Lambda/\ell_\Lambda$,
    \end{itemize}
  \item a set of configurations $\mathcal{Z}_\Lambda$ satisfying
    $\displaystyle\pro(\mathcal{Z}_\Lambda)\geq 1 -
    e^{-\log^{\alpha-2}|\Lambda|}$,
  \end{itemize}
  so that, for $L$ sufficiently large (depending only on
  $(\alpha,\alpha',\rho,\rho',\rho'')$), one has
  \begin{itemize}
  \item for $\omega\in\mathcal{Z}_\Lambda$, there exists at least
    $\displaystyle \frac{|\Lambda|}{\ell_\Lambda^d}\left(1+
      O\left(\log^{1/\rho'-1/\rho}|\Lambda|\right)\right)$ disjoint
    boxes $\Lambda_{\ell_\Lambda}(\gamma_j)$ satisfying the properties
    (1), (2) and (3) described in Theorem~\ref{thr:vsmall1} where
    $\ell'_\Lambda$ in~\eqref{error} satisfies~\eqref{eq:24};
  \item the number of eigenvalues of $H_\omega(\Lambda)$ that are not
    described above is bounded by
    \begin{equation*}
      C N(I_\Lambda)|\Lambda|\,\left[\log^{1/\rho'-1/\rho}
        |\Lambda|+\log^{1/(d\rho'')-1/(d\rho')}|\Lambda|\right].
    \end{equation*}
  \end{itemize}
\end{Th}
\noindent In Theorem~\ref{thr:vbig0}, our choice of parameters is made
so as to allow as small as possible density of states (see the second
condition in~\eqref{eq:9}). The price to pay for this is that the
width of the interval $I_\Lambda$ may not be optimal in all regimes
(compare the first condition in~\eqref{eq:9} with the width of the
intervals treated in~\cite{Ge-Kl:10} when $N(I_\Lambda)\gtrsim
|I_\Lambda|^{1+\rho}$). It is nevertheless essentially optimal when
$N(I_\Lambda)\asymp e^{-|I_\Lambda|^{-\rho}}$.\vskip.1cm\noindent
We also note that, in the proof of Theorem~\ref{thr:vbig0}, the choice
of $\ell=\ell_\Lambda$ guarantees that $|I_\Lambda|\ell^d\ll1$
(see~\eqref{eq:12}); thus, Lemma~\ref{lemasympt} gives a precise
description of:
\begin{itemize}
\item the probability distribution of the $\gamma$'s for which
  $H_\omega(\Lambda_\ell(\gamma))$ has exactly one eigenvalue in
  $I_\Lambda$,
\item the distribution of this eigenvalue when this is the case
\end{itemize}
\begin{proof}[Proof of Theorem~\ref{thr:vbig0}]
  We follow the proof of Theorem~\ref{thr:vbig1}. First, note that, as
  in the proof of Theorem~\ref{thr:vsmall1}, assumption~\eqref{eq:9}
  implies that
  \begin{equation}
    \label{eq:12}
    |I_\Lambda|\lesssim\log^{-1/\rho}|\Lambda|,\quad
    \log^{1/(d\rho')}|\Lambda|\lesssim \ell_\Lambda,\quad
    \log^{1/(d\rho'')}|\Lambda|\lesssim \ell'_\Lambda 
  \end{equation}
  With our choice of $\ell_\Lambda$ and $\ell'_\Lambda$, one estimates
  $K$ defined in~\eqref{defK} and $K'$ defined in~\eqref{defK'} by
  \begin{equation}
    \label{eq:10}
    \begin{aligned}
      \log^{\alpha+1/\rho'-1/\rho}|\Lambda|&\lesssim K\lesssim
      \log^{\alpha'+1/\rho'-1/\rho}|\Lambda| \\\quad\text{and}\quad
      \log^{\alpha+1/(d\rho')-1/(d\rho)}|\Lambda| &\lesssim K'\lesssim
      \log^{\alpha'+1/(d\rho')-1/(d\rho)}|\Lambda| .
    \end{aligned}
  \end{equation}
  In the present case, the computations done in~\eqref{estr2}
  and~\eqref{estr3} give a too gross estimate of the number of
  eigenvalues, or, equivalently, of the number of localization
  centers, on which one cannot get a precise control if one wants to
  keep a good probability estimate; this comes from the fact that the
  quantities $N(I_\Lambda)|\Lambda||I_\Lambda|$, $\ell$, $\ell'$ may
  be powers of $\log|\Lambda|$. We need to study more carefully the
  number of eigenvalues missed by the description constructed in the
  proof of Theorem~\ref{thr:vbig1}. Therefore, we
  follow the ideas used in the proof of~\cite[Theorem 4.1]{Kl:11b}. \\
  Let $\Gamma_{\ell,L}$ be the set of cubes
  $\{\Lambda_\ell(\gamma_j);\ j\}$ of the decomposition introduced in
  the proof of Theorem~\ref{thr:vbig1}. We partition $\Gamma_{\ell,L}$
  into $2^d$ sets such that, any two cubes $\Lambda_\ell(\gamma)$ and
  $\Lambda_\ell(\gamma')$ in each set, the Hamiltonians
  $H_\omega(\Lambda_\ell(\gamma))$ and
  $H_\omega(\Lambda_\ell(\gamma'))$ are independent. Let these sets be
  $(\Gamma_{\ell,L,j})_{1\leq j\leq 2^d}$ and their cardinality be
  $\tilde N_j:=\#\Gamma_{\ell,L,j}$. One has $\tilde N_j\asymp
  |\Lambda|\ell^{-d}$.\\
  For $\Lambda_{\ell}(\gamma_k)\in \Gamma_{\ell,L,j}$, set
  $X_{j,k}=\tr\car_{I_\Lambda}(H_\omega(\Lambda_\ell(\gamma_k)))$.
  These variables are i.i.d. and their common distribution is
  described by~\ref{thmWM}.\\
  We now want to estimate the maximal number of localization centers
  of $H_\omega(\Lambda)$ contained in boxes of
  $(\Gamma_{\ell,L})_{1\leq j\leq 2^d}$ that each contain at least two
  centers (these are the boxes of $\mathcal{S}_{\ell,L}$ in the
  notations of the proof of Theorem~\ref{thr:vbig1}). We want to show
  that this number is, with a probability close to $1$, bounded by $C
  K$ where $K$ is defined in~\eqref{defK} and $C>0$ is a constant to
  be chosen. \\
  Let $\P_r$ be the probability to have $2^{d+1}r$ localization
  centers of $H_\omega(\Lambda)$ in cubes containing at least two
  centers. We compute
  \begin{equation*}
    \begin{split}
      \P_r&\leq \sum_{j=1}^{2^d} \sum_{n_j=1}^{r} \binom{\tilde
        N_j}{n_j}\P\left(\min_{1\leq k\leq n_j}X_{j,k}\geq2\text{
          and }\sum_{k=1}^{n_j}X_{j,k}\geq 2r\right)\\
      &\leq \sum_{j=1}^{2^d}\sum_{n_j=1}^{r} \binom{\tilde N_j}{n_j}
      \sum_{\substack{l_1+\cdots+l_{n_j}\geq 2r\\ \forall 1\leq k\leq
          n_j,\ l_k\geq2}}\prod_{k=1}^{n_j}\P\left(
        X_{j,k}=l_k\right)\\
      &\leq \sum_{j=1}^{2^d}\sum_{n_j=1}^{r} \binom{\tilde N_j}{n_j}
      \sum_{\substack{l_1+\cdots+l_{n_j}\geq 2(r-n_j)\\ \forall 1\leq
          k\leq n_j,\ l_k\geq0}}\prod_{k=1}^{n_j}\P\left(
        X_{j,k}=l_k+2\right).
    \end{split}
  \end{equation*}
  Then, using~\eqref{HOM} and the independence of the local
  Hamiltonians associated to boxes in $\Gamma_{\ell,L,j}$, we obtain
  \begin{equation*}
    \begin{split}
      \P_r&\lesssim \sum_{j=1}^{2^d}\sum_{n_j=1}^{r} \binom{\tilde
        N_j}{n_j} \sum_{\substack{l_1+\cdots+l_{n_j}\geq 2(r-n)\\
          \forall 1\leq k\leq n_j,\ l_k\geq0}} \frac{(C
        N(I_\Lambda)\ell^d)^{n_j} (|I_\Lambda|\ell^d)^{2r-n_j}}
      {(l_1+1)!\cdots (l_{n_j}+1)!}  \\&\lesssim
      \sum_{j=1}^{2^d}\sum_{n_j=1}^{r} \frac1{n_j!}  (\tilde N_j
      N(I_\Lambda)\ell^d)^{n_j} (|I_\Lambda|\ell^d)^{2r-n_j}
      \\&\lesssim \sum_{n=1}^{r} \frac1{n!}
      (|\Lambda|N(I_\Lambda))^{n} (|I_\Lambda|\ell^d)^{2r-n}
      \\&\lesssim \left(K^\eta|I_\Lambda|\ell^d\right)^{r}+
      \left(\frac{CK}{\eta r}\right)^{r}.
    \end{split}
  \end{equation*}
  In the last estimate, we have cut the previous sum into two parts,
  the first when $n$ run from $0$ to $\eta r$ and the second from
  $\eta r$ to $r$. We now choose $\eta>0$ so that
  $\eta\alpha'<1/\rho-1/\rho'$ (for $\alpha'$ in~\eqref{eq:9}, see
  also~\eqref{eq:10}).
  Recall that $|I_\Lambda|\ell^d$ is small. Thus, setting $r=
  \eta^{-1}CK$ for some large $C>0$ and using~\eqref{eq:10}, we obtain
  \begin{equation}
    \label{eq:14}
    \P_r\leq e^{-K/C} \leq e^{-\log^{\alpha-1}|\Lambda|/C}.
  \end{equation}  
  Now, as in the proof of Theorem~\ref{thr:vbig1}, cover $
  \Lambda\setminus \cup_j \Lambda_\ell(\gamma_j)$ with a partition of
  boxes $\Lambda_{\ell'}(\gamma'_j)$ (see the proof of~\cite[Theorem
  1.16]{Ge-Kl:10} for more details). In the same way as above, one
  estimates the maximal number of centers of localization contained in
  the union of the boxes $\Lambda_{\ell'}(\gamma_j)\subset
  \Lambda$. Let $\P'_r$ be the probability that this number exceeds
  $r$. Then, for $r=CK'$ (for some constant $C>0$) where $K'$ is
  defined in~\eqref{defK'}, as above, we prove
  \begin{equation}
    \label{eq:15}
    \P'_r \leq e^{-K'/C} \leq e^{-\log^{\alpha-1}|\Lambda|/C}.
  \end{equation}  
  Summing~\eqref{eq:14} and~\eqref{eq:15}, taking into
  account~\eqref{eq:10} and~\eqref{eq:12}, we obtain that, for
  $\alpha$ sufficiently large and properly chosen
  $(\rho',\rho'')\in(d\rho,1)^2$, with probability at least
  $1-e^{-\log^2|\Lambda|}$, there are at most
  \begin{equation*}
    C(K+K')\lesssim
    N(I_\Lambda)|\Lambda|\,\left[\log^{1/\rho'-1/\rho}|\Lambda|
      +\log^{1/(d\rho'')-1/(d\rho')}|\Lambda|\right].
  \end{equation*}
  eigenvalues that are not accounted for by the description given in
  Theorem~\ref{thr:vbig0}. This completes the proof of
  Theorem~\ref{thr:vbig0}.
\end{proof}

\section{Applications to eigenvalue statistics: proofs}
\label{sec:applic}

\subsection{The proof of Theorem~\ref{ThLD}}
\label{sec:proof-theorem1}
First note that under the assumptions of Theorem~\ref{ThLD}, the
assumptions of Theorem~\ref{thr:vbig1} are fulfilled. We will use this
decomposition.\\
Let $X=X(\Lambda_\ell,I_\Lambda)$ the Bernoulli random variable equal
to 1 if $H_{\omega,\Lambda_\ell}$ has an eigenvalue in $I_\Lambda$ and
zero otherwise. Recall that the distribution of this random variable
is described by~\eqref{eq:51} in Lemma~\ref{lemasympt}.\vskip.1cm
\noindent To prove Theorem~\ref{ThLD}, we consider the collection of
Bernoulli random variables $X_j:=X(\Lambda_\ell(\gamma_j),I_\Lambda)$,
$j=1,\cdots, \tilde{N}$ defined by
\begin{itemize}
\item $X(\Lambda_\ell(\gamma),I_\Lambda)$ is defined in
  section~\ref{sec:distr-local-eigenv}
\item the boxes $(\Lambda_\ell(\gamma_j))_{1\leq j\leq \tilde N}$ are
  given by Theorem~\ref{thr:vbig1} and $\tilde
  N\asymp|\Lambda|/\ell^d$.
\end{itemize}
Thus, the random variables $(X_j)_j$ are i.i.d.\\
We have
\begin{equation}
  \label{decomp}
  \left| \tr \car_{I_\Lambda}(H_\omega(\Lambda))- N(I_\Lambda)|\Lambda|
  \right| \le \left| \tr \car_{I_\Lambda}(H_\omega(\Lambda)) -
    \sum_{j=1}^{\tilde{N}} X_j \right| + \left|
    \sum_{j=1}^{\tilde{N}} X_j - N(I_\Lambda)|\Lambda|
  \right| 
\end{equation}
By Theorem~\ref{thr:vbig1}, with a probability $\ge 1 -\exp\left(-c
  (N(I_\Lambda) |\Lambda|)^{\delta_{\kappa,\nu}} \right)$ we have
\begin{equation}
  \label{boundtr}
  \left| \tr \car_{I_\Lambda}(H_\omega(\Lambda)) - \sum_{j=1}^{\tilde{N}} X_j
  \right| \lesssim (N(I_\Lambda) |\Lambda| )^{\gamma_{\kappa,\nu}}.
\end{equation}
Next, the large deviation principle for i.i.d. $(0,1)$-Bernoulli
variables with expectation $p$ gives, for $\delta\in]\frac12, 1[$,
yields (see e.g.~\cite{MR98m:60001}),
\begin{align}
  \label{LDX}
  \P\left(\left|\sum_{j=1}^{\tilde{N}} X_j - p\tilde{N}\right| \ge
    (p\tilde{N})^ \delta \right) \le C \exp\left(-c_p (p\tilde{N})^{2
      \delta-1}\right),
\end{align}
where the constant $c_p$ is uniformly bounded as $p\downarrow 0$. We
apply the latter with $p=\P(X=1)$. On the account of
Lemma~\ref{lemasympt},~\eqref{condeq} and~\eqref{KK'bound}, we have
\begin{equation}
  \label{pN}
  |p\tilde{N}-N(I_\Lambda) |\Lambda| | \lesssim
  (N(I_\Lambda)|\Lambda|)^{1-\nu}.
\end{equation}
Combining~\eqref{decomp},~\eqref{boundtr},~\eqref{LDX} and~\eqref{pN}
we obtain, for any $\delta'\in]0,1[$,
\begin{multline}
  \P\left\{ \left| \tr \car_{I_\Lambda}(H_\omega(\Lambda))-
      N(I_\Lambda)|\Lambda| \right| \le
    (N(I_\Lambda)|\Lambda|)^{\gamma_{\kappa,\nu}} +
    (N(I_\Lambda)|\Lambda|)^{1-\nu} + (p\tilde{N})^{\frac12 + \frac12
      \delta'}\right\} \\ \ge 1- \exp\left(- (N(I_\Lambda) |\Lambda|
    )^{\delta_{\kappa,\nu}} \right) - C \exp\left(-c_p (p\tilde{N})^{
      \delta'}\right).
\end{multline}
Since $\gamma_{\kappa,\nu}>1-\nu$, and choosing
$\delta'=\delta_{\kappa,\nu}$, we get
\begin{multline}
  \P\left\{ \left| \tr \car_{I_\Lambda}(H_\omega(\Lambda))-
      N(I_\Lambda)|\Lambda| \right| \le
    3(N(I_\Lambda)|\Lambda|)^{\max(\gamma_{\kappa,\nu},\frac12 +
      \frac12 \delta_{\kappa,\nu})}  \right\} \\
  \ge 1- 2 \exp\left(- (N(I_\Lambda) |\Lambda| )^{\delta_{\kappa,\nu}}
  \right).
\end{multline}
Taking $\kappa$ sufficiently small, we get $\delta_{\kappa,\nu}=
\kappa\nu$. The result follows with $\eps=\kappa\nu$.\vskip.2cm
\noindent We turn to the proof of point~(2) which follows from the
previous analysis and the central limit theorem for Bernoulli random
variables, provided $\nu>\frac 12$.\\
This completes the proof of Theorem~\ref{ThLD}.
\subsection{The proofs of Theorems~\ref{thr:3},~\ref{thr:2}
  and~\ref{thr:1}}
\label{sec:proofs-theorems}
As already stated above, one can follow verbatim the proofs of the
corresponding results in~\cite{Ge-Kl:10,Kl:10a}. Let us just make a
few comments on those proofs.
\subsubsection{The proof of Theorem~\ref{thr:3}}
\label{sec:proof-theorem-1}
The corresponding result is~\cite[Theorem
1.2]{Ge-Kl:10}. In~\cite{Ge-Kl:10}, we also have the stronger Theorems
1.3 and 1.6 that also have their analogues in the present
setting. Comparing~\eqref{eq:60} with the
condition~\cite[(1.12)]{Ge-Kl:10}, we see that we have now only a
local condition at $E_0$ only. This gain is obtained thanks to
Lemma~\ref{lemasympt} and Theorem~\ref{thr:vsmall1}.
\subsubsection{The proofs of Theorem~\ref{thr:2} and~\ref{thr:1}}
\label{sec:proof-theorem-2}
The result corresponding to Theorem~\ref{thr:2} is~\cite[Theorem
1.4]{Kl:10a}. Theorem~\ref{thr:1} is a then a consequence of
Theorem~\ref{thr:2}, see e.g.~\cite{Mi:11,Kl:10a}. Under uniform
assuptions of the type $N(J)e^{|J|^{-\rho}}\geq1$ for some
$\rho\in(0,1)$ and for all $J\subset E_0+I_\Lambda$, one may as well
follow the method developed in~\cite{Ge-Kl:10}.
\\
As we shall see, to prove Theorem~\ref{thr:2} in the present case is
easier than in \cite[Theorem1.4]{Kl:10a}. The basic idea of the proof
of asymptotic ergodicity in~\cite[Theorem 1.4]{Kl:10a} is, for a given
interval $E_0+I_\Lambda$, to split it into smaller intervals such that, on
most of these intervals, the assumptions of a reduction of the same
type as Theorem~\ref{thr:vbig0} is valid plus a remaining set of
energies that only contains a negligible fraction of the eigenvalues
in $E_0+I_\Lambda$. Then, one proves asymptotic ergodicity for each of the
small intervals. Therefore, one needs to use an analogue of
Lemma~\ref{lemasympt} to control the eigenvalues. This imposes further
restrictions on how one has to choose the small intervals. In the
present case, thanks to the improvement obtained in
Lemma~\ref{lemasympt} over its analogues in~\cite[Lemma 2.2]{Ge-Kl:10}
and~\cite[Lemma 2.2]{Kl:10a}, the way to split the interval
$E_0+I_\Lambda$ will be much simpler (as a comparison of what follows with
the discussions following~\cite[Theorem 2.1 and Lemma 2.2]{Kl:10a} and
in~\cite[Section 3.2.1]{Kl:10a} will immediately show).\\
We will not give a complete proof of Theorem~\ref{thr:2} only
indicate the changes to be made in the proof of~\cite[Theorem
1.4]{Kl:10a}.\\
Pick $\alpha_0<\alpha<\alpha''<\alpha'$ (where $\alpha_0$ is given by
Theorem~\ref{thr:vbig0}). Pick $\mu>0$. Now, partition $I_\Lambda$
into intervals $(I_j)_{j\in J}$ such that
$N(I_j)\asymp|\Lambda|^{-1}\log^{\alpha''}|\Lambda|$. Define
\begin{equation}
  \label{eq:18}  B=\{j\in J;\ N(I_j)\leq|I_j|^{\mu}\}.
\end{equation}
Then, one clearly has
\begin{equation*}
  |I_\Lambda|\geq\sum_{j\in B}|I_j|\geq\sum_{j\in
    B}N(I_j)^{1/\mu}\asymp\# B |\Lambda|^{-1/\mu}\log^{\alpha''/\mu}|\Lambda|
\end{equation*}
thus, $\#B\lesssim |\Lambda|^{1/\mu}$ and
\begin{equation*}
  N\left(\bigcup_{j\in B}I_j\right)\lesssim |\Lambda|^{(1-\mu)/\mu}
  \log^{\alpha''}|\Lambda|.
\end{equation*}
The number of eigenvalues expected in $E_0+I_\Lambda$ is of order
$N(E_0+I_\Lambda)|\Lambda|$, thus, by~\eqref{eq:63}, larger than
$|\Lambda|^{1-\delta}$. Pick $\nu>0$. By the enhanced Wegner's
estimate~\eqref{W} and Markov's inequality, we know that, with
probability at least $1-|\Lambda|^{-\nu}$, the number of eigenvalues in
$\bigcup_{j\in B}I_j$ is bounded by
$|\Lambda|^{\nu+1/\mu}\log^{\alpha''}|\Lambda|$. We now pick $\nu$ 
and $\mu^{-1}$ small so that $\nu+1/\mu< 1-\delta$.\\
For $j\not\in B$, one has $N(I_j)\geq|I_j|^{\mu}$ and, thus, one can
then apply Theorem~\ref{thr:vbig0} to $I_j$ for $j\not\in B$. So, in
each $I_j$, we control $N(I_j)|\Lambda|(1+o(1))$ eigenvalues (the error is uniform in $j$ by
Theorem~\ref{thr:vbig0}); thus, the total number of eigenvalues we
control exceeds $\sum_{j\not\in B}N(I_j)|\Lambda|(1+o(1))$ that is,
exceeds $N(E_0+I_\Lambda)|\Lambda|(1+o(1))\geq|\Lambda|^{1-\delta}$
as announced above.\\
For $j\not\in B$, we moreover want to be able to apply
Lemma~\ref{lemasympt} to control the eigenvalues ``in'' the cubes
$\Lambda_{\ell_\Lambda}(\gamma)$ constructed in
Theorem~\ref{thr:vbig0} and the interval $I_j$. Therefore, we need to
check that $|I_j||\Lambda_{\ell_\Lambda}(\gamma)|\ll 1$
(see~\eqref{eq:51} and~\eqref{eq:52}). This is guaranteed
by~\eqref{eq:24} in
Theorem~\ref{thr:vbig0}.\\
Now thanks to Theorem~\ref{thr:vbig0} and Lemma~\ref{lemasympt}, in
each $I_j$ for $j\not\in B$, we reason as in the proof
of~\cite[Theorem 1.4]{Kl:10a}, or more precisely, as in the proof
of~\cite[Lemma 3.2]{Kl:10a} to obtain the asymptotic ergodicity.
\begin{Rem}
  \label{rem:3}
  One can actually prove Theorem~\ref{thr:2} on intervals to which the
  IDS gives a smaller weight, that is, relax assumption~\eqref{eq:63}
  into
  $|\Lambda|\cdot\log^{-\beta}|\Lambda|\cdot N(E_0+I_\Lambda)\to+\infty$
  for not too small $\beta$ (e.g. for not too negative $\beta$). Then,
  the condition $N(I_j)\leq|I_j|^\mu$ defining $B$ will have to be
  replaced with conditions of the type $N(I_j)\leq e^{-|I_j|^{-\rho}}$
  ($\rho$ will now be in $(0,1)$).\\
  Moreover, if $\beta$ is not sufficiently large, a restriction
  analogous to~\cite[(1.10) in Theorem 1.4]{Kl:10a} will come up
  again.
\end{Rem}

\def\cprime{$'$} \def\cydot{\leavevmode\raise.4ex\hbox{.}}
\def\cprime{$'$}

\end{document}